\documentclass[12pt]{iopart}

\usepackage{iopams}
\usepackage[utf8]{inputenc}
\usepackage{upgreek}
\usepackage{enumerate}
  \expandafter\let\csname equation*\endcsname\relax
  \expandafter\let\csname endequation*\endcsname\relax
\usepackage{amssymb,amsmath,amsthm}

\usepackage[colorlinks=true]{hyperref}
\hypersetup{urlcolor=blue, citecolor=blue}
\usepackage[usenames,dvipsnames,svgnames,table]{xcolor}

\newtheorem{thm}{Theorem}[section]
\newtheorem{cor}{Corollary}
\newtheorem{lem}[thm]{Lemma}
\newtheorem{pro}{Proposition}
\newtheorem{defn}[thm]{Definition}
\newtheorem{rem}{Remark}
\newtheorem{ex}{Example}

\newcommand{\s}{\mathcal{S}}
\newcommand{\Li}{\mathcal{L}}
\newcommand{\Liu}{\mathcal{L}_{\Upsilon}}
\newcommand{\si}{\mathbb{S}^1}

\begin{document}

\title[Simple representation of global second-order normal forms]
      {A simple global representation for second-order normal forms of Hamiltonian systems relative to periodic flows}

\author{M Avenda\~no-Camacho,
J A Vallejo\footnote{Permanent address:
Facultad de Ciencias, Universidad Aut\'onoma de San Luis Potos\'i,
Lat. Av. S. Nava s/n Col. Lomas, CP 78290 San Luis Potos\'i (SLP) M\'exico.}
and Yu Vorobjev}

\address{Departamento de Matem\'aticas, Universidad de Sonora,
Blvd L. Encinas y Rosales s/n Col. Centro,
Ed. 3K-1 CP 83000 Hermosillo (Son) M\'exico.}

\ead{misaelave@mat.uson.mx,jvallejo@fc.uaslp.mx,yurimv@guaymas.uson.mx}

\begin{abstract}
We study the determination of the second-order normal form for perturbed
Hamiltonians $H_{\epsilon}=H_0 +\epsilon H_1 +\frac{\epsilon^2}{2} H_2$,
relative to the periodic flow of the unperturbed Hamiltonian $H_0$.
The formalism presented here is global, and can be easily implemented in any CAS.
We illustrate it by means of two examples: the H\'enon-Heiles and the elastic pendulum Hamiltonians.
\end{abstract}
\pacs{02.40.Yy,45.10.Hj,45.10.Na}

\section{Introduction}
In this paper we discuss some computational aspects of the normal form theory for
Hamiltonian systems on general phase spaces, that is, Poisson manifolds.
According to Deprit
\cite{Dep82}, a perturbed vector field
$$A=A_0+\epsilon A_1+\frac{\epsilon^2}{2}A_2+\cdots +
\frac{\epsilon^k}{k!}A_k+O(\epsilon^{k+1})$$
on a manifold $M$, is said to be in normal form of order $k$ relative to $A_0$ if
$[A_0,A_i]=0$ for $i\in\{1,\ldots ,k\}$. In the context of perturbation theory, the
normalization problem is formulated as follows: to find a (formal or smooth) transformation
which brings a perturbed dynamical system to a normal form up to a given order.
The construction of a normalization transformation, in the framework of the Lie transform method
\cite{Dep69,Hor66,Kam70,Mey1}, is related to the solvability of a set of linear non homogeneous
equations, called the homological equations. If the homological equations admit global solutions, defined
on the whole $M$, we speak of a global normalization, which essentially depends on the properties of
the unperturbed dynamics.

Here we are interested in the global normalization of a perturbed Hamiltonian dynamics relative to
periodic Hamiltonian flows. In this case, a result due to Cushman \cite{Cus84}, states that if $A$
is Hamiltonian, and the flow of the unperturbed vector field $A_0$ is periodic, then the true dynamics admits
a global Deprit normalization to arbitrary order. The corresponding normal forms can be determined by a
recursive procedure (the so-called Deprit diagram) involving the resolution of the homological equations at
each step.

In this paper, we extend Cushman's result to the Poisson case and derive an alternative coordinate-free
representation for the second-order normal form, involving only three intrinsic operations: two averaging
operators associated to the $\si-$action, and the Poisson bracket. We give a direct derivation of this
representation based on a period-energy argument \cite{Gor69} for Hamiltonian systems, and some
properties of the periodic averaging on manifolds \cite{MVor-11,Cus84,Mos70}. This formalism allows us to
get an efficient symbolic implementation for some models related to polynomial perturbations of the
harmonic oscillator with $1:1$ resonance. In particular, we compute the second-order normal form
of the H\'enon-Heiles \cite{Cus84}, and the elastic pendulum \cite{BM81,BB73,Geo99} Hamiltonians, expressed in terms
of the Hopf variables.

Let us remark that the second-order normal form plays a very important r\^ole in the approximation of a perturbed
dynamics by solutions of the averaged system when a long-time scale is used \cite{AKN87,MSV07}. Our desire to study
this kind of dynamics led to the present work.

Sections \ref{sec2} and \ref{sec3} contain some basic properties of the action induced by the flow of
a periodic vector field and their associated averaging operators. In Section \ref{sec4} we particularize to
the case of Hamiltonian vector fields, using an energy-period relation, and the main result is proved in Section \ref{sec5}.
The final section is devoted to the examples.

\section{Vector fields with periodic flow}\label{sec2}

Throughout the paper, we set $\si =\mathbb{R}/2\pi\mathbb{Z}$. We collect here some results regarding the flow $\mathrm{Fl}^t_X$
of a vector field $X$, on an arbitrary manifold $M$,
in the case when $\mathrm{Fl}^t_X$ is periodic. Although these results are general, later they will be applied to the
case of a Hamiltonian vector field on a Poisson manifold $(M,P)$.

Let $X\in\mathcal{X}(M)$ be a complete vector field whose flow is periodic with period function
$T\in\mathcal{C}^{\infty}(M)$, $T>0$, that is:
for any $p\in M$,
\begin{equation}\label{eq1}
\mathrm{Fl}^{t+T}_X (p)=\mathrm{Fl}^{t}_X (p) .
\end{equation}
Then, $X$ determines an $\si -$action $\si \times M\rightarrow M$ given by $(t,p)\mapsto \mathrm{Fl}^{t/\omega (p)}_X (p)$,
where $\omega := 2\pi/T>0$ is the frequency function, and $t\in \si$. Thus, the $\si -$action is periodic, with constant period $2\pi$.

The generator $\Upsilon$ of this $\si-$action can be readily computed:
$$
\Upsilon (p)=\left. \frac{\mathrm{d}}{\mathrm{d}t}\right|_{t=0}\mathrm{Fl}^{t/\omega (p)}_X (p)
=\frac{1}{\omega (p)}\left. \frac{\mathrm{d}}{\mathrm{d}s}\right|_{s=0}\mathrm{Fl}^s_X (p)
=\frac{1}{\omega (p)}X (p),
$$
so $\Upsilon =\frac{1}{\omega}X$.
Notice, from \eqref{eq1}, that $T(p)>0$ is the period of the integral curve of $X$ passing through $p\in M$
at $t=0$, $c_p :\mathbb{R}\rightarrow M$ (which is such that $c(0)=p$ and
$\dot{c}_p(0)=X(p)$). In other words, $c_p(0)=p=c_p(T(p))$. Also, each point on the image of the
integral curve $c_p$, gives the same value for the period: $T(p)=T(c_p(t)),\mbox{ for all }t\in \mathbb{R}$.
In terms of the flow of $X$, that means
$$
((\mathrm{Fl}^t_X)^*T)(p)=T(\mathrm{Fl}^t_X(p))=T(p), \mbox{ for all }p\in M.
$$
As $T$ is constant along the orbits of $X$, its Lie derivative with respect to $X$ vanishes:
$$
\Li_X T=\left. \frac{\mathrm{d}}{\mathrm{d}t}\right|_{t=0}(\mathrm{Fl}^t_X)^*T=0.
$$
Now, from $T\omega =2\pi$, we get
$$
0=\Li_X (T\omega )=(\Li_X T)\omega +T\Li_X \omega =T\Li_X \omega.
$$
But $T>0$, so this implies that $\omega$ is a first integral (or invariant) of $X$,
\begin{equation}\label{eq2}
\Li_X \omega =0.
\end{equation}
\begin{defn}
A smooth function $f\in\mathcal{C}^{\infty}(M)$ is said to be an $\si -$invariant if it is
invariant under the flow of the generator $\Upsilon =\frac{1}{\omega}X$, that is,
$$
\Liu f=0.
$$
\end{defn}
\noindent Clearly, this is equivalent to the condition $(\mathrm{Fl}^t_{\Upsilon})^*f=f$, for all
$t\in [0,2\pi ]$.
Notice that, by \eqref{eq2}, the frequency function is also an invariant of the $\si-$action,
$\Liu \omega =\frac{1}{\omega}\Li_X \omega =0$.

\section{Averaging operators}\label{sec3}

Given a vector field $X\in\mathcal{X}(M)$ with periodic flow, the associated $\si-$action
can be used to define two averaging operators, which we will denote by $\langle \cdot \rangle$ and
$\s$. In this section, $M$ will be an arbitrary manifold.

For any tensor field $R\in \Upgamma T^s_r (M)$  ($r-$covariant, $s-$contravariant), the average of
$R$ with respect to the $\si-$action on $M$ induced by $X$, is the tensor field (of the same type as
$R$) defined by
$$
\langle R\rangle :=\frac{1}{2\pi}\int^{2\pi}_0(\mathrm{Fl}^t_{\Upsilon})^*R \,\mathrm{d}t.
$$
\noindent The properties of the flow \cite{AM88} guarantee that $\langle R\rangle$ is well-defined as a
differentiable tensor field.
Also, note that if $R\in \Upgamma T^s_r (M)$, and $X_1,\ldots,X_r\in\mathcal{X}(M)$, $\alpha_1,\ldots,\alpha_s
\in\Omega^1(M)$ are arbitrary, then, for every $p\in M$, $t\mapsto (\mathrm{Fl}^t_{\Upsilon})^*R(X_1,\ldots,X_r,\alpha_1,\ldots,\alpha_s )(p)$ is a real differentiable funcion on the compact $[0,2\pi]$, hence integrable. We will use this definition mainly applied to the case of functions $f\in \mathcal{C}^{\infty}(M)$
($(0,0)-$tensors) and vector fields $Y\in\mathcal{X}(M)$ ($(0,1)-$tensors).\\
The other averaging operator that will be important in what follows, is the $\s$ operator,
$\s :\Upgamma T^s_r (M)\rightarrow \Upgamma T^s_r (M)$. It is given by
$$
\s (R):=\frac{1}{2\pi}\int^{2\pi}_0(t-\pi)(\mathrm{Fl}^t_{\Upsilon})^*R \,\mathrm{d}t.
$$
\noindent Note that both, $\langle \cdot \rangle$ and $\s$, are $\mathbb{R}-$linear operators. Other properties are listed below.
\begin{lem}
For any complete vector field $Y\in\mathcal{X}(M)$ (whose flow is not necessarily periodic) and smooth
tensor field $R\in \Upgamma T^s_r (M)$, we have:
$$
\left. \frac{\mathrm{d}}{\mathrm{d}s}\right|_{s=0}(\mathrm{Fl}^s_Y)^*\langle R\rangle =
\frac{1}{2\pi}\left( (\mathrm{Fl}^{2\pi}_Y)^*R-R \right) ,
$$
where the averaging is taken with respect to the flow of $Y$, that is,
$\langle R\rangle$ is given by $\langle R\rangle :=\frac{1}{2\pi}\int^{2\pi}_0(\mathrm{Fl}^t_Y)^*R \,\mathrm{d}t $.
\end{lem}
\begin{proof}
Start from the identities (which follow directly from the definitions of flow and Lie derivative):
$$
(\mathrm{Fl}^t_Y)^*(\Li_Y R)= \frac{\mathrm{d}}{\mathrm{d}t}(\mathrm{Fl}^t_Y)^* R
=\left. \frac{\mathrm{d}}{\mathrm{d}s}\right|_{s=0}(\mathrm{Fl}^{s+t}_Y)^* R
=\left. \frac{\mathrm{d}}{\mathrm{d}s}\right|_{s=0}(\mathrm{Fl}^{s}_Y)^*(\mathrm{Fl}^{t}_Y)^*R .
$$
Taking the integral with respect to $t$ between $0$ and $2\pi$ on both sides, we get, on the one hand:
$$
\frac{1}{2\pi}\int^{2\pi}_0 (\mathrm{Fl}^t_Y)^*(\Li_Y R)\, \mathrm{d}t
= \left. \frac{\mathrm{d}}{\mathrm{d}s}\right|_{s=0}(\mathrm{Fl}^{s}_Y)^*\left(\frac{1}{2\pi}
\int^{2\pi}_0  (\mathrm{Fl}^t_Y)^* R\, \mathrm{d}t  \right)
=\left. \frac{\mathrm{d}}{\mathrm{d}s}\right|_{s=0}(\mathrm{Fl}^s_Y)^*\langle R\rangle,
$$
and, on the other:
$$
\frac{1}{2\pi}\int^{2\pi}_0 (\mathrm{Fl}^t_Y)^*(\Li_Y R) \, \mathrm{d}t
=\frac{1}{2\pi}\int^{2\pi}_0  \frac{\mathrm{d}}{\mathrm{d}t}(\mathrm{Fl}^t_Y)^*R \, \mathrm{d}t
=\frac{1}{2\pi}\left( (\mathrm{Fl}^{2\pi}_Y)^*R-R \right).
$$
\end{proof}
\begin{pro}\label{pro1}
For every $R\in \Upgamma T^s_r (M)$, the following properties hold:
\begin{enumerate}[(a)]
\item\label{ita} $R$ is invariant under the flow of $\Upsilon$ (that is, $\si-$invariant) if and only if
$\langle R\rangle =R$.
\item\label{itb} $\Liu \langle R\rangle =0$.
\item\label{itc} If $g\in\mathcal{C}^{\infty}(M)$ is $\si-$invariant, then $\langle gR\rangle =g\langle R\rangle $.
\item\label{itd} The averaging operator commutes with tensor contractions whenever one of the tensors is
$\si-$invariant, that is, if $S\in \Upgamma T^b_a (M)$ is $\si-$invariant and
$C^l_k$ is any contraction, then $\langle C^l_k(R\otimes S)\rangle =C^l_k(\langle R\rangle \otimes S)$.
\end{enumerate}
\end{pro}
\begin{proof}
\mbox{}
\begin{enumerate}[(a)]
\item If $R$ is invariant under the flow of $\Upsilon$, then $(\mathrm{Fl}^{t}_\Upsilon)^*R=R$, for all
$t\in[0,2\pi]$, and from this it is immediate that $\langle R\rangle =R$.
Reciprocally, if $\langle R\rangle =R$ we may apply the preceding lemma to obtain:
$$
\left. \frac{\mathrm{d}}{\mathrm{d}s}\right|_{s=0}(\mathrm{Fl}^s_\Upsilon)^*R
=\frac{1}{2\pi}\left( (\mathrm{Fl}^{2\pi}_\Upsilon)^*R-R \right),
$$
and from the fact that the flow of $\Upsilon$ is $2\pi-$periodic,
$$
\Liu R=\left. \frac{\mathrm{d}}{\mathrm{d}t}\right|_{t=0}(\mathrm{Fl}^t_\Upsilon)^*R=0.
$$
\item From the properties of the Lie derivative and the definition of $\langle R\rangle$:
$$
(\mathrm{Fl}^t_Y)^*(\Li_Y \langle R\rangle)= \frac{\mathrm{d}}{\mathrm{d}t}(\mathrm{Fl}^t_\Upsilon )^*\langle R\rangle
=\frac{\mathrm{d}}{\mathrm{d}t}\frac{1}{2\pi}\int^{2\pi}_0 (\mathrm{Fl}^{s+t}_\Upsilon)^* R\, \mathrm{d}s
=\frac{\mathrm{d}}{\mathrm{d}t}\frac{1}{2\pi}\int^{t+2\pi}_t (\mathrm{Fl}^{u}_\Upsilon)^* R\, \mathrm{d}u.
$$
Now, because $\mathrm{Fl}^{u}_\Upsilon$ is $2\pi-$periodic:
$$
(\mathrm{Fl}^t_Y)^*(\Li_Y \langle R\rangle)
=\frac{\mathrm{d}}{\mathrm{d}t}\frac{1}{2\pi}\int^{2\pi}_0 (\mathrm{Fl}^{u}_\Upsilon)^* R\, \mathrm{d}u=0,
$$
so, as $\mathrm{Fl}^{t}_\Upsilon$ is a diffeomorphism, $\Liu \langle R\rangle =0$.
\item It is a straightforward computation.
\item It is just a consequence of the commutativity between the pull-back and the tensor contractions, and the
functorial property $(\mathrm{Fl}^{t}_\Upsilon)^*(R\otimes S)=(\mathrm{Fl}^{t}_\Upsilon)^*R\otimes (\mathrm{Fl}^{t}_\Upsilon)^*S$.
\end{enumerate}
\end{proof}
\begin{rem}\label{nota}
In particular, from \eqref{itd} we get that if $Y\in\mathcal{X}(M)$ and $\alpha \in\Omega (M)$ is $\si-$invariant, then
$\langle i_Y \alpha \rangle =i_{\langle Y \rangle}\alpha$.
\end{rem}
\begin{pro}\label{pro2}
For any $R\in \Upgamma T^s_r (M)$ and $g\in\mathcal{C}^{\infty}(M)$ $\si-$invariant, the following hold:
\begin{enumerate}[(a)]
\item $\s (gR)=g\s (R)$.
\item $(\Liu \circ \s )(R)=R-\langle R\rangle $.
\end{enumerate}
\end{pro}
\begin{proof}
\mbox{}
\begin{enumerate}[(a)]
\item A straightforward computation.
\item With an obvious change of variable, we have:
$$
(\mathrm{Fl}^{s}_\Upsilon)^*\s (R)=\frac{1}{2\pi}\int^{2\pi}_0(t-\pi)(\mathrm{Fl}^{s+t}_{\Upsilon})^*R\,\mathrm{d}t
=\frac{1}{2\pi}\int^{s+2\pi}_s(u-s-\pi)(\mathrm{Fl}^{u}_{\Upsilon})^*R\,\mathrm{d}u.
$$
Differentiating both sides of this identity with respect to the parameter $s$, and taking into account
the $2\pi-$periodicity of the flow $\mathrm{Fl}^{s}_{\Upsilon}$, it results:
$$
\frac{\mathrm{d}}{\mathrm{d}s}(\mathrm{Fl}^{s}_\Upsilon)^*\s (R)=
(\mathrm{Fl}^{s}_\Upsilon)^*(R-\langle R\rangle ).
$$
The statement follows by recalling that $\mathrm{Fl}^{s}_\Upsilon$ is a diffeomorphism, and the identity (see \cite{AM88}):
$$
\frac{\mathrm{d}}{\mathrm{d}s}(\mathrm{Fl}^{s}_\Upsilon)^*\s (R)= (\mathrm{Fl}^{s}_\Upsilon)^*(\Liu \s (R)).
$$
\end{enumerate}
\end{proof}
Finally, let us give some useful properties involving the averaging operators.
\begin{pro}\label{pro3}
For all $R\in \Upgamma T^s_r (M)$, the operators $\Liu$, $\langle \cdot \rangle$, and $\s$, satisfy the relations:
\begin{enumerate}[(a)]
\item\label{itea} $\langle \Liu R \rangle =\Liu \langle R \rangle =0$.
\item $\langle \s (R) \rangle = \s (\langle R\rangle ) =0$.
\item\label{itec} $\langle \mathrm{d}\alpha \rangle =\mathrm{d}\langle \alpha \rangle $, for all $\alpha \in\Omega (M)$.
\end{enumerate}
\end{pro}
\begin{proof}
Straightforward computations, making use of Proposition \ref{pro2} and the fact that $\mathrm{d}$ commutes
with pull-backs.
\end{proof}

\section{The Hamiltonian case}\label{sec4}

Let $(M,P)$ be an $m-$dimensional Poisson manifold, where $P\in \Upgamma \Uplambda^2 TM$ is a Poisson
bivector determining a bracket $\{ f,g\}=P(\mathrm{d}f,\mathrm{d}g)$, for all $f,g\in\mathcal{C}^{\infty}(M)$.
For every $f$, its Hamiltonian vector field $X_f\in\mathcal{X}(M)$ is given by
$X_f (g):=\{ f,g\}$, for any $g\in\mathcal{C}^{\infty}(M)$, equivalently,
\begin{equation}\label{hamv}
X_f =i_{\mathrm{d}f}P .
\end{equation}
At any point the distribution spanned by the Hamiltonian vector fields is involutive, as a consequence
of Jacobi's identity for the Poisson bracket $\{ \cdot ,\cdot \}$. Thus, these Hamiltonian vector fields give rise
to a foliation whose leaves turn out to be symplectic manifolds (see \cite{Wei83}). On each leaf $S$, the
restriction $P|_S$ is a non-degenerate Poisson bivector field which determines a symplectic structure
$\sigma_S$ through:
$$
\sigma_S (X_f,X_g):=\{ f,g\}.
$$
Indeed, by the splitting theorem due to Weinstein (\cite{Wei83}), the local structure of $(M,P)$ can be described as follows:
for any $p\in M$ there exists a chart $(U,\phi )$ of $M$ around $p$ such that, if
$\{ q_1,...,q_k,p_1,...,p_k,y_1,...,y_l \}$ are the coordinates of $\phi:U\rightarrow \mathbb{R}^m$
($2k+l=m$), then
\begin{equation}\label{poisson}
P|_U =\sum^k_{i=1}\frac{\partial}{\partial q_i}\wedge\frac{\partial}{\partial p_i}
+\frac{1}{2}\sum^l_{i,j=1}\varphi_{ij}(y_1,...,y_l)\frac{\partial}{\partial y_i}\wedge\frac{\partial}{\partial y_j},
\end{equation}
where $\varphi :\pi_l (U)\subset \mathbb{R}^l\rightarrow \mathbb{R}$ is smooth and $\varphi_{ij}(p)=0$
($\pi_l :\mathbb{R}^m=\mathbb{R}^{2k}\times \mathbb{R}^l\rightarrow \mathbb{R}^l$ is the canonical projection).
The non-negative integer $k$ is called the rank of the Poisson structure $P$ at $p\in M$. When $k=m$, $P$ induces
a symplectic structure on $M$. Then, the symplectic leaf $S$ through $p\in M$, is given by the equations $(y_1,...,y_l)=(0,...,0)$.

When moving along the flow of a Hamiltonian vector field, which is tangent to some integral submanifold $S$,
it is clear that we stay on the same symplectic leaf $S$. Next, we study what happens on these leaves when
the Hamiltonian vector field has periodic flow.

We will need first an auxiliary result, interesting in its own, known as the period-energy relation (see
\cite{Gor69}).
\begin{pro}\label{pro4}
Let $X$ be a vector field on the symplectic manifold $(S,\sigma )$ whose flow is periodic with
period function $T\in\mathcal{C}^{\infty}(M)$, $T>0$ (and frequency $\omega=2\pi/T$). If $X$ is the Hamiltonian
vector field of a certain function $f\in\mathcal{C}^{\infty}(M)$ (that is, $i_X \sigma =-\mathrm{d}f$), then:
\begin{equation}\label{enper}
\mathrm{d}\omega \wedge \mathrm{d}f=0=\mathrm{d}T\wedge\mathrm{d}f.
\end{equation}
\end{pro}
\begin{proof}
By hypothesis, we have,
$$
\Li_X \sigma =i_X \mathrm{d}\sigma +\mathrm{d}i_X \sigma =-\mathrm{d}^2 f=0.
$$
On the other hand, using the generator $\Upsilon =X/\omega$ of the $\si-$action induced by $X$:
$$
\Li_X \sigma =\omega \Liu \sigma -\frac{1}{\omega}\mathrm{d}\omega\wedge \mathrm{d}f.
$$
Recalling that $\omega ,f$ are first integrals of $X$, and hence $\si-$invariants, applying the
averaging operator $\langle \cdot \rangle$ to the last identity, taking into account that
$\langle \Liu \sigma\rangle =0$ (Proposition \ref{pro3} \eqref{itea}), and the commutativity
between $\mathrm{d}$ and $\langle \cdot \rangle$  (Proposition \ref{pro3} \eqref{itec}), we get:
$$
0=\langle \omega\Liu \sigma \rangle - \langle \frac{1}{\omega}\mathrm{d}\omega\wedge \mathrm{d}f \rangle
=\omega\langle \Liu \sigma \rangle -\frac{1}{\omega}\mathrm{d}\omega\wedge \mathrm{d}f
=-\frac{1}{\omega}\mathrm{d}\omega\wedge \mathrm{d}f.
$$
\end{proof}
\begin{rem}
Notice that, in terms of Hamiltonian vector fields, we can
write the energy-period relation \eqref{enper} as follows,
\begin{equation}\label{enper2}
X_\omega \wedge X_f =0.
\end{equation}
\end{rem}
Also, in the course of the proof we have seen that, if $\Upsilon =\frac{1}{\omega}X$ is the generator
of the $\si-$action induced by $X$:
$$
0=\Li_X \sigma =\omega\Liu \sigma -\frac{1}{\omega}\mathrm{d}\omega\wedge \mathrm{d}f,
$$
so from \eqref{enper} we get the following consequence.
\begin{cor}
The symplectic form $\sigma$ is $\si-$invariant, $\Liu \sigma =0$. In particular, $\langle \sigma \rangle =\sigma$.
\end{cor}
Notice that, under the hypothesis of Proposition \ref{pro4},
if $g\in\mathcal{C}^{\infty}(S)$ is $\si-$invariant, then its Hamiltonian vector
field $X_g\in\mathcal{X}(S)$ is also $\si-$invariant.
Indeed, recalling that $\mathrm{d}$ commutes with the averaging (Proposition \ref{pro3} \eqref{itec}),
Remark \ref{nota}, and the preceding Corollary, we get:
$$
i_{X_{\langle g \rangle}}\sigma =-\mathrm{d}\langle g \rangle =-\langle \mathrm{d}g \rangle
=\langle i_{X_{g}}\sigma \rangle =i_{\langle X_g \rangle }\sigma .
$$
Hence, by the non-degeneracy of $\sigma$, $\langle X_g \rangle = X_{\langle g \rangle}$. Now, if
$g$ is $\si-$invariant, $\langle g \rangle =g$, and so $\langle X_g \rangle =X_g$.\\
As a consequence, for any $\si-$invariant $g\in\mathcal{C}^{\infty}(S)$, we have
$$
\Liu X_g =[X_{\Upsilon},X_g ]=0.
$$

Now, suppose that we are given a function $H\in\mathcal{C}^{\infty}(M)$ on the Poisson manifold
$(M,P)$ such that its Hamiltonian vector field $X_H \in\mathcal{X}(M)$ has periodic flow
(with frequency function $\omega\in \mathcal{C}^{\infty}(M)$, $\omega >0$). Let
$\Upsilon =\frac{1}{\omega}X_H$ be the generator of the associated $\si-$action. From the
results above we know that $M$ is foliated by symplectic leaves $S$ in such a way that
$P|_S$ is equivalent to a symplectic form $\sigma_S$ (recall \eqref{poisson}), and these
are invariant under Hamiltonian flows. Thus:
$$
0=\Li_{X_H}P=\Li_{\omega \Upsilon}P=
\omega\Liu P-\frac{\omega\Upsilon}{\omega}\wedge i_{\mathrm{d}\omega}P=
\omega\Liu P+\frac{1}{\omega} X_H \wedge X_\omega ,
$$
where we have used the formula $\Li_{fX}A=f\Li_X A-X\wedge i_{\mathrm{d}f}A$ (valid for
any function $f\in\mathcal{C}^{\infty}(M)$, vector field $X\in\mathcal{X}(M)$ and multivector
field $A\in\Upgamma (\Uplambda TM)$, see \cite{Mic08}, p. 358), as well as \eqref{hamv} and the fact that $\omega >0$.
From this identity and the energy-period
relation \eqref{enper2}, we deduce that $P$ is $\si-$invariant, $\Liu P=0$.\\
Moreover, if  $g\in\mathcal{C}^{\infty}(M)$ is $\si-$invariant, the flow of its Hamiltonian vector
field $X_g$ leaves the integral submanifolds $S$ invariant and, as we have seen, on each of them
it satisfies $\Liu X_g=0$, so this is also true on $M$. In other words, the flows of $\Upsilon$ and
$X_g$ commute on $M$. The following result exploits this fact.
\begin{pro}
Let $(M,P)$ be a Poisson manifold, and $H\in\mathcal{C}^{\infty}(M)$ such that its Hamiltonian
vector field $X_H \in\mathcal{X}(M)$ has periodic flow. If $f,g\in\mathcal{C}^{\infty}(M)$ and
$g$ is $\si-$invariant, then:
\begin{enumerate}[(a)]
\item\label{item0} If $\omega$ is the frequency function of $X_H$, then, $X_H \wedge X_\omega =0$.
\item\label{itema} $\{H,g\}=0$.
\item\label{itemb} $\langle \{f,g\}\rangle =\{\langle f \rangle ,g\}$.
\end{enumerate}
\end{pro}
\begin{proof}
Item \ref{item0} follows from the above considerations, while \eqref{itema} is proved by a straightforward computation. Item \eqref{itemb} is a direct
consequence of the $\si-$invariance of $g$ and the fact that the flows of $\Upsilon$ and
$X_g$ commute.
\end{proof}

\section{The main result}\label{sec5}
Let $H_\varepsilon =H_0 +\varepsilon H_1+\frac{1}{2}\varepsilon^2H_2 +O(\varepsilon^3)$ an
$\varepsilon -$dependent Hamiltonian function which describes a perturbed Hamiltonian system on a
Poisson manifold $(M,P)$, with associated bracket $\{ \cdot ,\cdot \}$. We will denote by
$X_{H_\varepsilon}=X_{H_0}+\varepsilon X_{H_1}+\frac{1}{2}\varepsilon^2 X_{H_2}+O(\varepsilon^3)$
the corresponding Hamiltonian vector field.
Recall that the perturbed Hamiltonian vector field $X_{H_\varepsilon}$ is in (Deprit) normal form
relative to $X_{H_0}$ of order $k$ in $\varepsilon$ if
\begin{equation}\label{eq3}
[X_{H_0},X_{H_i}]=0,\mbox{ for all }i\in\{ 1,2,...,k \}.
\end{equation}
\noindent In terms of Hamiltonian functions, \eqref{eq3} is satisfied whenever
$$
\{ H_0,H_i \}=0,\mbox{ for all }i\in\{ 1,2,...,k \}.
$$

Usually, one can bring the Hamiltonian to a normal form by means of near-to-identity transformations. Let us
recall some definitions and basic properties.

Let $M$ be a manifold, $N\subset M$ be a non-empty open domain, and $\delta >0$. A smooth mapping
$\Phi :(-\delta ,\delta )\times N\to M$ is said to be a near-to-identity transformation if, for
each $\varepsilon\in (-\delta,\delta )$, the map $\Phi_\varepsilon :N\to M$ given by
$$
\Phi_\varepsilon (x)=\Phi (\varepsilon,x)
$$
is such that it is a diffeomorphism onto its image and, moreover,
$
\Phi_0 =\mathrm{id}_M$.

These transformations have the following important property: whenever we have a time-dependent vector field
$A_\varepsilon$ on $M$, and a near-to-identity transformation $\Phi_\varepsilon$, the pull-back
$\Phi^*_\varepsilon A_\varepsilon$ is again an $\varepsilon-$dependent vector field on $N$, and it is
such that,
$$
\left. \Phi^*_\varepsilon A_\varepsilon \right|_{\varepsilon =0}=A_0.
$$
In other words, thinking of $A_\varepsilon$ as a perturbed vector field, near-to-identity transformations
preserve the unperturbed part.

Actually, we will construct the required transformations out from the flow of a perturbed vector field.
The following properties say that we can do that on each open domain with compact closure.
\begin{pro}
Let $F:\mathbb{R}\times M\to M$ a smooth mapping, sending $(\varepsilon,x)$ to $F_\varepsilon (x)=F(\varepsilon,x)$,
such that $F_0=\mathrm{id}_M$. Then, for any open domain with compact closure $N\subset M$, there exists
a $\delta >0$ such that, for each $\varepsilon\in (-\delta,\delta)$, the restriction
$F_\varepsilon |_N$ is a diffeomorphism onto its image.
\end{pro}
\begin{proof}
It is an immediate consequence of the fact that the closure $\overline{N}$
can be covered by a finite number of open neighborhoods, such that the Implicit Function Theorem applies on them.
\end{proof}
\begin{pro}
Let $A_\varepsilon =A_0 +\varepsilon R_\varepsilon$ be a smooth vector field on a manifold $M$. Assume that the unperturbed vector field
$A_0$ is complete on $M$. Then, for any open domain $N\subset M$, with compact closure, and any constant $\delta >0$, there
exists another constant $L>0$ such that the flow $\mathrm{Fl}^t_{A_\varepsilon}$ of $A_\varepsilon$, is well-defined on $N$ for any
$t\in [0,L/\varepsilon ]$ and each $\varepsilon \in (0,\delta ]$.
\end{pro}
\begin{proof}
If $X,Y$ are vector firlds on the manifold $M$, their flows are related by
\begin{equation}\label{auxiliar}
\mathrm{Fl}^t_X \circ \mathrm{Fl}^t_{P_t}=\mathrm{Fl}^t_Y,
\end{equation}
where $P_t$ is the time-dependent vector field given by $P_t =-X+(\mathrm{Fl}^t_X)^*Y$. Now, let
$$
(\mathrm{Fl}^t_{A_0})^*A_\varepsilon -A_0 =\varepsilon R_t (\varepsilon),
$$
where $R_t (\varepsilon )=(\mathrm{Fl}^t_{A_0})^*R_\varepsilon$ depends smoothly on $t$ and $\varepsilon$, and
fix a $\delta >0$. By the Flow-Box Theorem and the compactness of the closure $\overline{N}$, there exists an
$L>0$ such that the flow of $R_t (\varepsilon)$ is well-defined on $N$ for any $t\in [0,L]$. Applying \eqref{auxiliar}
to $X=A_0$, $Y=A_\varepsilon$, and $P_t =R_t (\varepsilon)$, we get,
$$
\mathrm{Fl}^t_{A_\varepsilon}=\mathrm{Fl}^t_{A_0}\circ \mathrm{Fl}^{t\varepsilon}_{R_t(\varepsilon)},
$$
and, since $\mathrm{Fl}^t_{A_0}$ is well-defined for all $t\in\mathbb{R}$, the statement follows.
\end{proof}
\begin{defn}
We say that the system described by a vector field of the form $A_\varepsilon =A_0 +\varepsilon R_\varepsilon$,
where $A_0$ has complete flow,
admits a global normalization of order $k$ if, for each open domain $N\subset M$ with compact closure, there
exist a $\delta >0$ and a near-to-identity transformation $F:(-\delta ,\delta )\times N\to M$, which brings $A_\varepsilon$
to a normal form of order $k$.
\end{defn}
\begin{thm}\label{mainthm}
Suppose that the flow of $X_{H_0}$ is periodic with frequency function $\omega\in\mathcal{C}^{\infty}(M)$,
$\omega>0$.
Then, the perturbed Hamiltonian system admits a global normalization of arbitrary order $k$.
In particular, the second order normal form can be expressed as:
\begin{equation}\label{heps}
H_\varepsilon \circ \Phi_\varepsilon = H_0+\varepsilon \langle H_1 \rangle + \frac{\varepsilon^2}{2}\left(\langle H_2 \rangle +\langle\{S\left( \frac{H_1}{\omega}\right),H_1 \}  \rangle \right)+O(\varepsilon^3).
\end{equation}
\end{thm}
\begin{proof}
If the Hamiltonian vector field $X_{H_0}$ has periodic flow, the existence of the near-to-identity
canonical transformation $\Phi_\varepsilon$ follows from the above Propositions (see also \cite{MVor-11,Cus84,Mey1,Mey2}). Here we give a
explicit formula for it.\\
Let $\Phi_\varepsilon$ be the flow of the perturbed vector field $Z_\varepsilon=Z_0+\varepsilon Z_1$ where $Z_0$ and
$Z_1$ are the Hamiltonian vector field of the functions $G_0=\frac{1}{\omega}\mathcal{S}(H_1)$ and $G_1=\frac{1}
{\omega}\mathcal{S}(H_{2}+\{\mathcal{S}(\frac{1}{\omega}H_1),H_1+\langle H_1\rangle\})$, respectively.
Using the Lie transform method \cite{Cus84,Dep69,Hor66,Kam70}, the second order development of
$H_\varepsilon \circ \Phi_\varepsilon$ is given by:
\begin{eqnarray}
H_\varepsilon \circ \Phi_\varepsilon &=& H_{0}+\varepsilon\left(
\mathcal{L}_{Z_0}H_{0}+H_{1}\right)\nonumber \\
&&+\frac{\varepsilon^{2}}{2}\left(  \mathcal{L}_{Z_0}^{2}H_{0}%
+2\mathcal{L}_{Z_0}H_{1}+\mathcal{L}_{Z_1}H_{0}+H_{2}\right)
+O(\varepsilon^{3})\label{Hnorm}
\end{eqnarray}
Now, we apply the results of the preceding sections to put this Hamiltonian in the form \eqref{heps}.
To this end, we compute:
\begin{align*}
\mathcal{L}_{Z_0}H_{0} &= -\mathcal{L}_{X_{H_0}}\mathcal{S}(\frac{1}{\omega}H_1)= \langle H_1 \rangle - H_1,\\
\mathcal{L}^2_{Z_0}H_{0} &= \mathcal{L}_{X_{G_0}}(\langle H_1 \rangle - H_1)=\{\frac{1}{\omega}\mathcal{S}(H_1), \langle H_1 \rangle - H_1\},\\
\mathcal{L}_{Z_0}H_{1} &= \mathcal{L}_{X_{G_0}}H_1=\{\frac{1}{\omega}\mathcal{S}(H_1), H_1\},
\end{align*}
and, finally
\begin{align*}
\mathcal{L}_{Z_1}H_{0} &= -\mathcal{L}_{X_{H_0}}\mathcal{S}(H_{2}+\{\mathcal{S}(\frac{1}{\omega}H_1),H_1+\langle H_1\rangle\})\\
&= \langle H_{2}\rangle+\langle\{\mathcal{S}(\frac{1}{\omega}H_1),H_1\}\rangle-(H_{2}+\{\mathcal{S}(\frac{1}{\omega}H_1),H_1+\langle H_1\rangle\}).
\end{align*}
Substituting these identities into \eqref{Hnorm}, we obtain the normal form \eqref{heps}.
\end{proof}

\section{Examples}\label{part2}
In this section we illustrate the computation of the normal form of two particular Hamiltonians on $\mathbb{R}^2$
endowed with the canonical symplectic form,
$\Omega = \mathrm{d}p_1\wedge \mathrm{d}q_1+\mathrm{d}p_2\wedge \mathrm{d}q_1$
(and the corresponding canonical Poisson bracket).
If we have a system admitting an $\mathbb{S}^1-$action, described by a perturbed
Hamiltonian $H=H_0 + \epsilon H_1$, and such that the Hamiltonian vector field of
$H_0$, $X_{H_0}$, has periodic flow with frequency $\omega$ then, as shown in
Theorem \ref{mainthm}, its second-order normal form is given by:
$$
    H_0 +\epsilon\langle H_1\rangle
    +\frac{\epsilon^2}{2}\left( \langle\lbrace \mathcal{S}(\frac{H_1}{\omega}),H_1 \rbrace \rangle \right).
$$
%
\begin{ex}[H\'enon-Heiles Hamiltonian]
This example is taken from \cite{Cus93}. The Hamiltonian is
$$
K=K_0 + \epsilon K_1 =\frac{1}{2}(p^2_1 +p^2_2)+\frac{1}{2}(q^2_1 +q^2_2)+\epsilon \left( \frac{q^3_1}{3}-q_1q^2_2 \right)
$$
(note that the perturbation term is an homogeneous polynomial of degree $3$).
The frequency function for the flow of $X_{K_0}$ is readily found to be constant, $\omega =1$,
and, after some computations, the second-order normal form is found to be:
\begin{align*}
&\frac{{p_2}^{2}+{p_1}^{2}}{2}+\frac{{q_2}^{2}+{q_1}^{2}}{2}-\frac{\epsilon^2}{48}\left( 5\,{q_2}^{4}+\left( 10\,{q_1}^{2}+10\,{p_2}^{2}-18\,{p_1}^{2}\right) \,{q_2}^{2}\right. \\
&\left. +56\,p_1\,p_2\,q_1\,q_2+5\,{q_1}^{4}+\left(10\,{p_1}^{2}-18\,{p_2}^{2}\right)\,{q_1}^{2}+5\,{p_2}^{4}+10\,{p_1}^{2}\,{p_2}^{2}+5\,{p_1}^{4} \right)
\end{align*}
It is usual to express the normal form in terms of the Hopf variables
$w_1,w_2,w_3,w_4$, as a previous step to carry on the reduction of symmetry
process (see \cite{Cus84},\cite{Cus93}).
For the case in which $H_0$ is the Hamiltonian of the $2D-$harmonic oscillator,
these variables form a system of functionally independent generators of the algebra
of first integrals of $H_0$, and are defined as $w_1=2(q_1q_2+p_1p_2)$, $w_2=2(q_1p_2-q_2p_1)$,
$w_3=q_1^2+p_1^2-q_2^2-p_2^2$, $w_4=q_1^2+q_2^2+p_1^2+p_2^2$. Working separately with
the independent term and the coefficient of $\epsilon^2$ in the expression above, we get:
$$
\frac{{w}_{4}}{2},
$$
and
$$
\frac{{w}_{2}^{2}\,\left( 48\,\lambda +7\right) }{48}-\frac{{w}_{4}^{2}\,\left( 48\,\lambda+5\right) }{48}+{w}_{3}^{2}\,\lambda+{w}_{1}^{2}\,\lambda .
$$
In the process of expressing the $q_i,p_i$ variables in terms of the $w_j$, a parameter $\lambda$ appears as a consequence of the fact that the corresponding system of equations is indeterminate.
The formulas appearing in \cite{Cus93} are recovered by choosing the value
$0$ of the parameter:
$$
\frac{7\,{w}_{2}^{2}}{48}-\frac{5\,{w}_{4}^{2}}{48}.
$$
Thus, the second-order normal form of the H\'enon-Heiles system is
$$
H_\epsilon \circ \Phi_\epsilon =\frac{w_4}{2}+\frac{\epsilon^2}{48}\left( 7w^2_2 -5w^2_4 \right) +O(\epsilon^3).
$$
\end{ex}

\begin{ex}[The elastic pendulum]
Consider the case of the Hamiltonian of a elastic pendulum (see \cite{BB73},\cite{BM81},\cite{Geo99}):
$$
     H(q_1,p_1,q_2,p_2)=\frac{p_1^2+p_2^2}{2}+\frac{q_1^2+q_2^2}{2}-\frac{\epsilon}{2} q_1^2(1+q_2),
$$
which is that of a perturbed system $H_0+\epsilon H_1$, where
$$
               H_0(q_1,p_1,q_2,p_2)=\frac{p_1^2+p_2^2}{2}+\frac{q_1^2+q_2^2}{2},
$$
and
$$
                    H_1(q_1,p_1,q_2,p_2)=-\frac{q_1^2(1+q_2)}{2}.
$$
Note that the perturbation term now is \emph{not} homogeneous. The computation of the normal form in the original variables gives the result:
\begin{align*}
&\frac{{p_2}^{2}+{p_1}^{2}}{2}+\frac{{q_2}^{2}+{q_1}^{2}}{2}-\frac{\epsilon}{4}\left( {q_1}^{2}+{p_1}^{2}\right) \\
&-\frac{{\epsilon}^{2}}{192}\,\left( \left( 20\,{q_1}^{2}-4\,{p_1}^{2}\right) \,{q_2}^{2}+48\,p_1\,p_2\,q_1\,q_2+5\,
{q_1}^{4}+ \right. \\
& \left. \left( -4\,{p_2}^{2}+10\,{p_1}^{2}+12\right) \,{q_1}^{2}+20\,{p_1}^{2}\,{p_2}^{2}+5\,{p_1}^{4}+12\,
{p_1}^{2}\right)
\end{align*}
As before, we can express in terms of the Hopf variables the independent terms and the coefficient of $\epsilon$,
getting:
$$
\frac{{w}_{4}}{2},
$$
and
$$
-\frac{{w}_{4}}{8}-\frac{{w}_{3}}{8}.
$$
Note, however, that the coefficient of $\epsilon^2$ is \emph{not} a homogeneous
polynomial (of degree $4$): there are two $2-$degree terms:
$(q_1^2+p_1^2)/16$. Luckily, these terms
can be easily expressed in terms of the variables $w_1,w_2,w_3,w_4$
(as $(q_1^2+p_1^2)/16=(w_4+w_3)/32$) and then we can analyse the remainder, which \emph{is}
a polynomial of degree $4$. Again, a parameter $\mu$ appears in the process:
$$
-\frac{{w}_{4}^{2}\,\left( 768\,\mu +25\right) }{768}+\frac{{w}_{3}^{2}\,\left( 256\,\mu +5\right) }{256}+\frac{{w}_{2}^{2}\,\left( 32\,\mu +1\right) }{32}+{w}_{1}^{2}\,\mu -\frac{5\,{w}_{3}\,{w}_{4}}{384}.
$$
Let us take the simplest solution $\mu =0$:
$$
-\frac{25\,{w}_{4}^{2}}{768}-\frac{5\,{w}_{3}\,{w}_{4}}{384}+\frac{5\,{w}_{3}^{2}}{256}+\frac{{w}_{2}^{2}}{32}.
$$
The remainder in the coefficient of $\epsilon^2$ is:
$$
\frac{{w}_{4}}{32}+\frac{{w}_{3}}{32}.
$$
Thus, we get the second-order normal form of the elastic pendulum
in the Hopf variables:
\[
H_\epsilon \circ \Phi_\epsilon =\frac{w_4}{2} -\frac{\epsilon}{8}\left( w_4+w_3\right)
+\frac{{\epsilon}^{2}}{32}\left( w_4+w_3+w^2_2-\frac{25{w_4}^{2}}{24}-\frac{5w_3w_4}{12}+\frac{5{w_3}^{2}}{8}\right)
+O(\epsilon^3).
\]
\end{ex}
\begin{rem}
One of the advantages of the representation \eqref{heps} for the second-order normal form,
is that it allows an easy implementation in any Computer Algebra System (CAS), as it does not involve
the resolution of the homological equations. Indeed, the computations above were
carried out with a package written in Maxima \cite{Max12}, available
at the URL \url{http://galia.fc.uaslp.mx/~jvallejo/pdynamics.zip}. It contains a detailed documentation
illustrating its use with the preceding examples.
\end{rem}

\end{document}